\documentclass[
reprint,
superscriptaddress,
 amsmath,amssymb,
 aps,
 pra,
]{revtex4-1}

\usepackage{lipsum}
\usepackage{graphicx}
\usepackage{dcolumn}
\usepackage{bm}
\usepackage{soul}
\usepackage{placeins}
\usepackage{amssymb}
\usepackage{amsthm}
\usepackage[utf8]{inputenc}
\usepackage[english]{babel}
\usepackage{mathtools}
\usepackage{tcolorbox}
\usepackage{tikz}
\usetikzlibrary{quantikz2}
\usepackage{braket}
\usepackage{multirow}

\newtheorem{theorem}{Theorem}
\newtheorem{lemma}{Lemma}
\newtheorem{definition}{Definition}
\newtheorem{corollary}[theorem]{Corollary}
\newtheorem{proposition*}{Proposition}

\DeclareMathOperator{\Tr}{Tr}

\usepackage[colorlinks=true,
            linkcolor=blue,
            citecolor=magenta,
            urlcolor=cyan]{hyperref}

\begin{document}

\title{Maximal coherence of quantum measurement and the resource theory of sharpness}

\author{Kyunghyun Baek}
\email{k.baek@yonsei.ac.kr}
\affiliation{Institute for Convergence Research and Education in Advanced Technology, Yonsei University, Seoul 03722, Republic of Korea}
\affiliation{Department of Quantum Information, Graduate School, Yonsei University, Seoul 03722, Republic of Korea}

\author{Yonggi Jo}
\affiliation{Agency for Defense Development, Daejeon 34186, Republic of Korea}

\author{Hyunchul Nha}
\email{hnha@utep.edu}
\affiliation{Department of Physics, University of Texas at El Paso, El Paso, Texas 79968, USA}

\date{\today}

\begin{abstract}
A resource theory of quantum measurement can be addressed in terms of quantum coherence and measurement sharpness, respectively. The former analyzes the off-diagonal structure of POVM elements in a predetermined basis while the latter analyzes  the deviation from trivial, state-independent, measurements. We establish a direct connection between the two resource theories by identifying measurement sharpness as the maximal coherence that is achievable under all possible unitary changes of the reference basis. For a broad class of POVMs whose elements share a common eigenbasis, we show that the maximal distance-based coherence of measurement coincides exactly with the corresponding distance-based sharpness monotone. We further extend this equivalence, with element-additive distances, to POVMs whose elements admit a common mutually unbiased basis structure. These results provide a measurement-theoretic analogue of the maximal-coherence \& purity correspondence for quantum states. We also show that the maximal coherence of measurement is faithful with respect to trivial measurements and is monotonic under fuzzifying operations for dichotomic measurements, as well as under mixed-unitary and unitarily covariant preprocessing channels. Finally, we illustrate the operational meaning and limitations of the equivalence through qubit POVMs, single-photon phase sensing, and noisy photon-number resolving detection. In particular, the maximal Fisher information in a Mach-Zehnder interferometer is shown to be determined by the squared maximal coherence of the measurement, while in an imperfect photon-number resolving detector the maximal coherence behaves as a proper sharpness monotone, unlike conventional PVM-based unsharpness measures.
\end{abstract}

\maketitle

\section{Introduction}\label{sec:intro}


Quantum systems can provide advantages over their classical counterparts in computation~\cite{Nielsen_Chuang_2010, Preskill2021}, metrology~\cite{Giovannetti2011,Paris2009}, and cryptography~\cite{Pironio2009}. These advantages originate from genuinely non-classical features of quantum states, channel and quantum measurements all together. The framework of quantum resource theories~\cite{Chitambar2019} provides a unified way to characterize and quantify such features by specifying which objects are free, which operations preserve the free set, and which functions serve as resource monotones.

Building on the resource theories of state-based quantities such as entanglement~\cite{Horodecki2009} and coherence~\cite{Baumgratz2014,Streltsov2017}, which have been developed over past decades, the resource-theoretic framework has been extended to quantum channels~\cite{Theurer2019,Gour2019,Gour2020,Gour2021,Xu2019} and measurements~\cite{Baek2020,Tendick2023,Buscemi2024}. General frameworks have also been developed based on robustness measures~\cite{Takagi2019}, and relationships among distinct resources such as entanglement, coherence, and relevant operational meanings have been actively explored~\cite{Streltsov2015, Kwon2019, Theurer2020, Ahnefeld2022}. Within this line of research, two intrinsic properties of a positive operator-valued measure (POVM) have been identified as resources: its \emph{coherence}~\cite{Baek2020}, which captures the off-diagonal structure of the POVM elements with respect to a chosen reference basis, and its \emph{sharpness}~\cite{Tendick2023,Buscemi2024,Busch2004,Carmeli2007,Heinosaari2008,Busch2009,Baek2016,Liu2021,Mitra2022,Mitra2025}, which quantifies how decisively the POVM discriminates orthogonal states and how close it is to an ideal projective observable. The aim of the present work is to establish a precise quantitative relationship between these two quantities.


The resource-theoretic framework has been established~\cite{Tendick2023} by using a general distance-based geometric quantification for measurement resources, which showed that any contractive, jointly convex, distance between POVMs leads to a valid monotone, and provided concrete examples based on the diamond norm and the Schatten $p$-norms. In particular, a complete and operational resource theory of measurement sharpness has been constructed~\cite{Buscemi2024}, identifying the trivial (uninformative) POVMs as the free objects and the \emph{fuzzifying operations}---preprocessing by a CP unital map followed by convex mixing with a trivial measurement---as the free operations. These works provide a well-defined framework in which sharpness can be quantified.

A natural next step is to ask whether the two measurement-based resources of coherence and sharpness are quantitatively related. The two notions capture different aspects of a POVM. Coherence is \emph{basis-dependent}, since rotating the incoherent basis by a unitary $U$ generically changes the degree of coherence, while sharpness is \emph{basis-independent}. Although a hierarchy between measurement resources has been reported~\cite{Tendick2023}, a direct quantitative connection between the two has not been developed. On the state side, an analogous connection was provided by Streltsov \emph{et~al.}~\cite{Streltsov2018}, who showed that the basis-optimized coherence of a quantum state recovers the purity. This maximal-coherence/purity correspondence motivates the study of an analogous relation on the measurement side, namely, whether the basis-optimized coherence of a measurement likewise recovers its sharpness.


In this work, we establish the relationship between coherence and sharpness of measurement for a wide class of POVMs. Define the \emph{maximal coherence of measurement} of $\mathsf{A}$ as
\begin{equation}\label{eq:CDmax_intro}
  C_{D}^{\max}(\mathsf{A})
  \;=\;\max_{U}\,C_{D}(U^{\dagger}\mathsf{A}\,U),
\end{equation}
which is the basis-optimized value of any distance-based coherence monotone $C_D$ over all unitaries $U$ in the underlying Hilbert space. Our main result is that $C_{D}^{\max}$ coincides with the corresponding distance-based sharpness $S_D$ of Ref.~\cite{Tendick2023} whenever the elements of $\mathsf{A}$ share a common eigenbasis, and---for element-additive distances---whenever they admit a common basis mutually unbiased to the incoherent basis. Outside this regime, e.g., for a symmetric, informationally complete (SIC) qubit POVM, the two quantities are inequivalent. We probe these behaviors in three physically motivated settings: a qubit measurement interpolating between a joint $X$--$Y$ POVM and a tetrahedral SIC POVM, where the gap between $C_D^{\max}$ and $S_D$ opens up as the common-MUB structure breaks.

Furthermore, we explore the monotonicity of the maximal coherence of measurement under fuzzifying operations. For dichotomic measurements, we show its monotonicity and for non-dichotomic measuremenet we find operationally relavent classes of channels such as mixed unitary channels and unitarily covariant channels. We show its monotonic behaviour in a Mach--Zehnder interferometer read out by an on--off detector, where the maximal Fisher information is directly related to the maximal coherence of measurement; and a noisy photon-number--resolving detector (PNRD), where $C_{D}^{\max}$ remains monotonic with respect to the quantum efficiency while conventional PVM-based unsharpness measures do not.

The paper is organized as follows. Section~\ref{sec:prelim} reviews the resource theories of coherence and sharpness of measurements. Section~\ref{sec:maxcoh} introduces the maximal coherence of measurement and presents our main results on its equivalence with sharpness and its monotonicity under fuzzifying operations. Section~\ref{sec:examples} illustrates these results with three physical examples: a qubit measurement, a Mach--Zehnder interferometer, and a noisy photon-number--resolving detector. Section~\ref{sec:conclusion} concludes.

\section{Preliminaries}\label{sec:prelim}

\subsection{Coherence of measurements}\label{subsec:coh}

A quantum measurement on a $d$-dimensional Hilbert space with $n$ outcomes
is described by a POVM, i.e., a tuple $\mathsf{A}=\{A_x\}_{x=1}^n$ of
positive semidefinite operators satisfying the completeness relation
$\sum_{x=1}^n A_x = I$. If every POVM element is a projection
($A_x^2 = A_x$), the POVM is called a projection-valued measure (PVM).

In the resource theory of state coherence~\cite{Baumgratz2014,Streltsov2017},
a quantum state is incoherent (with respect to a fixed reference basis $\{|i\rangle\}_{i=1}^{d}$) if it can be written as a statistical mixture of basis projectors, $\sigma_{\rm inc}=\sum_{i=1}^{d}p_i|i\rangle\langle i|$. Whenever a state has nonzero coherence, a quantum measurement should have nonzero coherence to detect or exploit it; this role is most directly seen in state discrimination tasks~\cite{Takagi2019,Takagi2019_2}.

In the same spirit, a quantum measurement is called \emph{incoherent} when all its elements are diagonal in the reference basis~\cite{Baek2020}:
\begin{align}
A_{x}=\sum_{i=1}^{d}\lambda_{i|x}\,|i\rangle\langle i|.
\end{align}
The set of incoherent measurements with $n$ outcomes on a $d$-dimensional system is denoted by $\mathcal{I}(d,n)$.

In a resource theory, free operations are operations that map free objects to free objects. For state coherence, the most general such class is that of \emph{maximally incoherent operations} (MIO), which preserve the set of incoherent states~\cite{AbergMIO,Streltsov2017}. For coherence of measurements, the natural analogue is the set of \emph{maximally incoherent operations for measurements} (MIO-M), defined as those quantum channels $\mathcal{E}$ whose Heisenberg dual maps incoherent measurements to
incoherent measurements:
\begin{align}
\mathcal{E}^{\dagger}(\mathcal{I}(d,n)) \subseteq \mathcal{I}(d,n).
\end{align}
We emphasize that MIO and MIO-M are not equivalent: MIO can in general \emph{create} coherence in measurements when acting on incoherent ones in
the Heisenberg picture, and conversely MIO-M can create coherence in states~\cite{Baek2020}. In what follows, $\mathcal{E}^{\dagger}$ always denotes the Heisenberg dual of a CPTP map $\mathcal{E}$, so that
$\Tr[\mathcal{E}^{\dagger}(A)\rho]=\Tr[A\,\mathcal{E}(\rho)]$.

A natural family of coherence monotones for measurements arises from a distance $D$ between POVMs as
\begin{align}
C_{D}(\mathsf{A})=\min_{\mathsf{B}\in\mathcal{I}(d,n)}D(\mathsf{A},\mathsf{B}),
\end{align}
where $D$ is required to satisfy faithfulness
$D(\mathsf{A},\mathsf{B})=0\Leftrightarrow\mathsf{A}=\mathsf{B}$ and contractivity
$D(\mathsf{A},\mathsf{B})\geq D(\mathcal{E}^{\dagger}(\mathsf{A}),
 \mathcal{E}^{\dagger}(\mathsf{B}))$ for any CP unital map $\mathcal{E}^{\dagger}$.
Together these imply, for any MIO-M operation,
\[
\begin{array}{cll}
\text{(i)} & [\text{Faithfulness}] & C_D(\mathsf{A})=0
   \;\Leftrightarrow\;\mathsf{A}\in\mathcal{I}(d,n),\\
\text{(ii)} & [\text{Monotonicity}] &
   C_D(\mathsf{A})\geq C_D(\mathcal{E}^{\dagger}_{\rm MIO-M}(\mathsf{A})),
\end{array}
\]
which are the standard properties for any convex resource theory~\cite{Chitambar2019} (see Appendix~\ref{App:1} for a self-contained proof). Since MIO-M is the most general set of free operations for any resource theory of coherence of measurements, every MIO-M monotone is also a monotone for any more restrictive resource theory.

A canonical way to construct such a distance between POVMs is via a state-induced statistical distance. Given a statistical distance $D$ on probability distributions, define
\begin{equation}\label{eq:Dstat}
D(\mathsf{A},\mathsf{B})=\sup_{\rho}D\big(\mathbf{p}_{\rho,\mathsf{A}},
\mathbf{p}_{\rho,\mathsf{B}}\big),
\end{equation}
where $\mathbf{p}_{\rho,\mathsf{A}}=\{\Tr[A_x\rho]\}_{x=1}^{n}$. By construction, this maximal statistical distance inherits faithfulness and
contractivity from $D$, and it is jointly convex whenever $D$ is jointly convex (see Appendix~\ref{App:1}, Lemma~\ref{join_convex}). Since a quantum measurement can equivalently be described as a quantum-to-classical channel mapping a state to its outcome distribution, the framework above is a special case of channel divergences. In particular, several relative-entropy
based channel divergences~\cite{Cooney2016,Leditzky2018,Gour2019_ieee,
Gour2019_prl} fit into this formalism, are jointly
convex~\cite{Nielsen_Chuang_2010}, and provide operationally meaningful notions of distinguishability between quantum measurements.

A particularly useful distance with a clear operational meaning is the \emph{diamond distance}. For two quantum channels $\mathcal{N}$ and $\mathcal{M}$,
\begin{align}
D_{\diamond}(\mathcal{N},\mathcal{M})
=\sup_{\rho_{AR}}\frac{1}{2}\big\|(\mathcal{N}\otimes I_R)(\rho_{AR})
-(\mathcal{M}\otimes I_R)(\rho_{AR})\big\|_1,
\end{align}
which equals the optimal single-shot probability advantage in a channel discrimination task. Representing each measurement as a measure-and-prepare channel $\Lambda_{\mathsf{A}}(\rho)=\sum_x\Tr[A_x\rho]\,|x\rangle\langle x|$,
the diamond distance between two measurements reduces to~\cite{Tendick2023}
\begin{align}\label{eq:diamond}
D_\diamond(\mathsf{A},\mathsf{B})
=\frac{1}{2}\,\max_{\rho_{AR}}\sum_x\big\|\sigma_{x}(\rho_{AR})
-\tau_{x}(\rho_{AR})\big\|_1,
\end{align}
where $\sigma_{x}(\rho)=\mathrm{Tr}_1[({A}_{x}\otimes I)\rho]$ and 
$\tau_{x}(\rho)=\mathrm{Tr}_1[({B}_{x}\otimes I)\rho]$, and $\Tr_1[\cdot]$ denotes the trace with respect to the first subsystem \cite{Tendick2023}. 
In this setting, the diamond norm is effectively reduced to a maximal trace distance between the conditional output states, averaged over measurement outcomes. 
By adopting the diamond distanc, one obtains a distance-based coherence monotone with a clear operational interpretation. The associated diamond-norm coherence monotone is
\begin{align}
C_\diamond(\mathsf{A})=\min_{\mathsf{B}\in\mathcal{I}(d,n)}
D_\diamond(\mathsf{A},\mathsf{B}),
\end{align}
which is faithful, contractive, and jointly convex as a consequence of the corresponding properties of the diamond norm~\cite{Tendick2023}.

\subsection{Sharpness of measurements}\label{subsec:sharp}

A useful operational view is that any POVM is an attempt to realize an ideal observable, possibly degraded by classical or quantum noise. Projective measurements provide the prototype of an ideal, repeatable
measurement, and a POVM that is a PVM is therefore called \emph{PVM-sharp}~\cite{Baek2016}.

\begin{definition}[PVM-sharpness]
$\mathsf{A}=\{A_x\}_{x=1}^n$ is PVM-sharp if and only if it is a PVM, that
is, $A_x^2=A_x$ for all $x$.
\end{definition}

Several quantitative measures of \emph{PVM-unsharpness} have been proposed. An information-theoretic quantifier was introduced in~\cite{Baek2016} as the sum of the von Neumann entropies of the POVM elements,
\begin{align}
U_{\rm ent}(\mathsf{A})=\sum_{x=1}^n S(A_x),\quad S(A_x)=-\Tr[A_x\log A_x].
\end{align}
A variance-based unsharpness was developed in~\cite{Liu2021}: the unsharpness matrix
$F_\rho(\mathsf{A})_{xy}=\delta_{xy}\Tr[\rho A_x]
-\Tr\!\left[\rho\,\tfrac{A_xA_y+A_yA_x}{2}\right]$
captures deviations from projectivity through second-order moments, and on the maximally mixed state one defines
\begin{equation}
U_{\rm var}(\mathsf{A})=\big\|F_{I/d}(\mathsf{A})\big\|_1,
\end{equation}
which vanishes if and only if $\mathsf{A}$ is a PVM.

An operator-norm based measure was introduced in~\cite{Mitra2022} as
\begin{equation}
U_{\rm op}(\mathsf{A})=\big\|I-E_{\mathsf{A}}\big\|,\qquad
E_{\mathsf{A}}=\sum_{x}A_{x}^{2},
\end{equation}
where $\|\cdot\|$ is the operator norm. Since $E_{\mathsf{A}}=I$ if and only if $\mathsf{A}$ is a PVM, $U_{\rm op}$ also vanishes precisely on PVMs.

While these measures all faithfully detect deviation from a PVM, none of them is monotonic under the natural class of free operations of a resource theory of sharpness. The reason is that PVM-sharpness does not distinguish between projective measurements and trivial (uninformative) ones whose elements are projectors of rank zero or full rank---both are formally PVMs. A more discriminating notion was therefore proposed in~\cite{Buscemi2024}.

\begin{definition}[Generalized sharpness]\label{def:gensharp}
$\mathsf{A}=\{A_x\}_{x=1}^n$ is generalized-sharp if for every outcome $x$
there exists a normalized vector $|\psi_x\rangle$ such that
$A_x|\psi_x\rangle=|\psi_x\rangle$.
\end{definition}

Equivalently, every effect $A_x$ has eigenvalue $1$. POVMs that do not hold this condition are called \emph{unsharp}. The completely unsharp (or trivial, uninformative) measurements are those of the form
$\mathsf{A}=\{p_x I\}_{x=1}^n$ with $\sum_x p_x = 1$; their outcome
statistics are independent of the measured state, so they extract no information about it~\cite{Tendick2023,Skrzypczyk2019}. The set of trivial measurements with \$n\$ outcomes on a \$d\$-dimensional system is denoted $\mathcal{T}(d,n)$. 

A resource theory of sharpness based on Definition~\ref{def:gensharp} was developed in~\cite{Tendick2023,Buscemi2024}. The free objects are the trivial measurements, and the free operations are required to map any trivial measurement to a trivial measurement. The Heisenberg dual $\mathcal{E}^{\dagger}$ of any CPTP map $\mathcal{E}$ is a CP unital linear map and therefore preserves the form of a trivial measurement element-wise, $\mathcal{E}^{\dagger}(p_x I_A)=p_x I_B$. However, no CP unital map can interconvert two distinct trivial measurements with different distributions $\{p_x\}$. For all trivial measurements with the same outcome set to be equivalent in the resource theory of sharpness, Buscemi \emph{et~al.}~\cite{Buscemi2024} introduced \emph{fuzzifying operations} $\mathcal{F}$, defined element-wise as
\begin{align}\label{FO}
\mathcal{F}:\;A_{x}\;\mapsto\;\alpha\,\mathcal{E}^{\dagger}(A_{x})+(1-\alpha)\,q_{x}\,I_{B},\qquad \forall x\in\{1,\dots,n\},
\end{align}
where $\mathcal{E}^{\dagger}:\mathcal{B}(\mathcal{H}_{A})\to\mathcal{B}(\mathcal{H}_{B})$ is the dual of an arbitrary CPTP quantum preprocessing channel $\mathcal{E}:\mathcal{B}(\mathcal{H}_{B})\to\mathcal{B}(\mathcal{H}_{A})$ between (possibly distinct) Hilbert spaces $\mathcal{H}_{A}$ and $\mathcal{H}_{B}$, $\alpha\in[0,1]$, and $\{q_{x}\}_{x=1}^{n}$ is a probability distribution. By construction, fuzzifying operations preserve the number of outcomes $n$ but \emph{may change the Hilbert space dimension}; in particular, dimension-preserving fuzzifying operations arise when $\mathcal{H}_{A}=\mathcal{H}_{B}$.

In this convex resource theory~\cite{Chitambar2019}, distance-based sharpness monotones are
\begin{equation}\label{eq:D_sharp}
S_D(\mathsf{A})=\min_{\mathsf{B}\in\mathcal{T}(d,n)}D(\mathsf{A},\mathsf{B}).
\end{equation}
Faithfulness and contractivity of $D$, together with joint convexity, imply
\[
\begin{array}{cll}
\text{(i)} & [\text{Faithfulness}] & S_D(\mathsf{A})=0
   \;\Leftrightarrow\;\mathsf{A}\in\mathcal{T}(d,n),\\
\text{(ii)} & [\text{Monotonicity}] & S_D(\mathsf{A})\geq
   S_D(\mathcal{F}(\mathsf{A})),
\end{array}
\]
as proved in~\cite{Tendick2023} and reviewed in Appendix~\ref{App:1}, Lemma~\ref{lem:sharp_mono}. Adopting the diamond norm yields the diamond-norm sharpness monotone
\begin{equation}\label{eq:D_sharp_diamond}
S_\diamond(\mathsf{A})=\min_{\mathsf{B}\in\mathcal{T}(d,n)}
D_\diamond(\mathsf{A},\mathsf{B}),
\end{equation}
which is the canonical sharpness monotone we will use throughout the rest of this work~\cite{Tendick2023}.

\section{Maximal coherence of measurements}\label{sec:maxcoh}

\subsection{Definition and faithfulness}

Coherence is intrinsically basis dependent: rotating the reference basis by a unitary $U$ alters the off-diagonal structure of every POVM element. We define the \emph{maximal coherence of measurement} of $\mathsf{A}$ as the
basis-optimized value of any coherence monotone $C_D$,
\begin{align}
C_{D}^{\max}(\mathsf{A})=\max_{U}C_{D}\!\left(U^{\dagger}\mathsf{A}\,U\right),
\end{align}
where the maximum is taken over all unitaries $U$ acting on the $d$-dimensional Hilbert space and $U^{\dagger}\mathsf{A}\,U$ denotes the POVM $\{U^{\dagger}A_x U\}_{x=1}^{n}$. Equivalently, $C_{D}^{\max}$ is the
coherence of $\mathsf{A}$ optimized over all reference bases. The
goal of this section is to show that, for an operationally relevant class of measurements, $C_{D}^{\max}$ coincides with the corresponding distance-based sharpness $S_D$, and that $C_{D}^{\max}$ behaves as a valid sharpness monotone.

We begin by establishing faithfulness.

\begin{theorem}[Faithfulness]\label{thm:MCfaithful}
The maximal coherence of measurement vanishes if and only if the
measurement is trivial,
\begin{align}
C_{D}^{\max}(\mathsf{A})=0\;\Longleftrightarrow\;
\mathsf{A}=\{p_{x}I\}_{x=1}^{n}
\end{align}
for some probability distribution
$\{p_x\}$.
\end{theorem}

\begin{proof}
The sufficiency is straightforward, since the identity is invariant
under any unitary conjugation, so $\{p_x I\}_{x=1}^n$ is incoherent in every
basis.

To prove necessity, we consider the contrapositive: assume
$\mathsf{A}\notin\mathcal{T}(d,n)$, i.e., at least one element $A_{x}$ is not
proportional to $I$. Spectrally decompose this element as
\begin{align*}
A_x=\sum_{i=1}^{d}\lambda_{i|x}|\lambda_{i|x}\rangle\langle\lambda_{i|x}|,
\end{align*}
where the eigenvalues $\lambda_{i|x}\geq 0$ are not all equal. Pick a
unitary $V$ that maps the eigenbasis $\{|\lambda_{i|x}\rangle\}$ to the
mutually unbiased basis (MUB) $\{|\lambda_{i|x}^{+}\rangle\}$ with respect to
the incoherent basis $\{|i\rangle\}$, namely
$|\langle i|\lambda_{j|x}^{+}\rangle|=1/\sqrt{d}$ for all $i,j$, with
$|\lambda_{j|x}^{+}\rangle=\frac{1}{\sqrt{d}}\sum_{k=0}^{d-1}\omega^{jk}|k\rangle$
and $\omega=e^{2\pi i/d}$. The off-diagonal elements of $V^{\dagger}A_xV$ in
the incoherent basis are
\begin{align*}
\langle k|V^{\dagger}A_xV|l\rangle
&=\frac{1}{d}\sum_{j=1}^{d}\lambda_{j|x}\,\omega^{j(k-l)},\qquad k\neq l,
\end{align*}
which is the discrete Fourier transform of $(\lambda_{1|x},\dots,
\lambda_{d|x})$. Since the $\lambda_{j|x}$ are not all equal, this Fourier
transform has at least one nonvanishing component, hence at least one
off-diagonal element is nonzero. Consequently
$V^{\dagger}\mathsf{A}V\notin\mathcal{I}(d,n)$ and any faithful coherence
monotone $C_D$ obeys $C_D(V^{\dagger}\mathsf{A}V)>0$, which gives
$C_D^{\max}(\mathsf{A})\geq C_D(V^{\dagger}\mathsf{A}V)>0$.
\end{proof}

\subsection{Equivalence with sharpness}


The connection between maximal coherence and sharpness becomes particularly transparent for POVMs whose elements are simultaneously diagonalizable. We show below that, in this case, the maximal coherence of measurement \(C_D^{\max}\) coincides with the distance from the set of trivial
measurements \(\mathcal{T}(d,n)\), and therefore exactly recovers the corresponding distance-based sharpness \(S_D\).

\begin{theorem}[Maximal coherence equals sharpness]\label{thm:MCsharp}
Suppose all POVM elements of $\mathsf{A}=\{A_x\}_{x=1}^n$ share a common eigenbasis, $A_x=\sum_{i=1}^{d}p_{x|i}|\lambda_i\rangle\langle\lambda_i|$. Then for any distance-based maximal coherence $C_D^{\max}$  induced by a distance that is invariant under simultaneous unitary conjugation of POVMs, as in Eqs.~\eqref{eq:Dstat} and \eqref{eq:diamond},
\begin{equation}\label{eq:MC_simplified}
C_{D}^{\max}(\mathsf{A})
=\min_{\sum_{x}q_{x}=1}D\!\left(\mathsf{A},\{q_{x}I\}_{x=1}^{n}\right)
\;=\;S_D(\mathsf{A}).
\end{equation}
\end{theorem}

\begin{proof}
We prove the two inequalities
\begin{align}\label{eq:ineqs}
C_{D}^{\max}(\mathsf{A})\leq S_{D}(\mathsf{A})
\quad\text{and}\quad
C_{D}^{\max}(\mathsf{A})\geq S_{D}(\mathsf{A}).
\end{align}

\emph{Upper bound.} Every trivial measurement is incoherent, so
\begin{align}
C_{D}^{\max}(\mathsf{A})
&=\max_{U}\min_{\mathsf{B}\in\mathcal{I}(d,n)}D(U^{\dagger}\mathsf{A}\,U,\mathsf{B})\\
&\leq\;\max_{U}\min_{\sum_{x}q_{x}=1}D\!\left(U^{\dagger}\mathsf{A}\,U,
\{q_{x}I\}\right)\\
&=\max_{U}\min_{\sum_{x}q_{x}=1}D\!\left(\mathsf{A},\{q_{x}I\}\right)
=S_{D}(\mathsf{A}),
\end{align}
where the second equality follows from the unitary invariance of the maximal statistical distance, $D(U^{\dagger}\mathsf{A}\,U,\{q_{x}I\})
=D(\mathsf{A},\{q_{x}I\})$, since $\{q_{x}I\}$ is itself unitarily invariant.

\emph{Lower bound.} Choose the unitary $V$ that maps the common eigenbasis $\{|\lambda_i\rangle\}$ to the mutually unbiased basis $\{|\lambda_i^{+}\rangle\}$ with respect to the incoherent basis. Then
\begin{align}
C_{D}^{\max}(\mathsf{A})
&\geq\min_{\mathsf{B}\in\mathcal{I}(d,n)}
D(V^{\dagger}\mathsf{A}\,V,\mathsf{B})\\
&\geq\;\min_{\mathsf{B}\in\mathcal{I}(d,n)}
D\!\left(\mathcal{E}_{+}^{\dagger}(V^{\dagger}\mathsf{A}\,V),
\mathcal{E}_{+}^{\dagger}(\mathsf{B})\right),
\end{align}
where the first inequality uses that $V$ is a particular (not necessarily optimal) choice of unitary, and the second uses contractivity of $D$ under the dephasing CP unital map $\mathcal{E}_{+}^{\dagger}$ in the MUB, $\mathcal{E}_{+}^{\dagger}(X)=\sum_{i}\langle\lambda_i^{+}|X|\lambda_i^{+}
\rangle\,|\lambda_i^{+}\rangle\langle\lambda_i^{+}|$.

By construction, $\mathcal{E}_+^{\dagger}$ leaves $V^{\dagger}\mathsf{A}V$ invariant (since this POVM is already diagonal in the MUB by conditions on $\mathsf{A}$ and the choice of $V$). Moreover, applied to any incoherent measurement $\mathsf{B}\in\mathcal{I}(d,n)$, every off-diagonal entry in the MUB is averaged to a constant, so $\mathcal{E}_+^{\dagger}(B_x)=\frac{1}{d}\Tr[B_x]\,I=q_x I$ with $q_x=\Tr[B_x]/d$, i.e., $\mathcal{E}_+^{\dagger}(\mathsf{B})\in\mathcal{T}(d,n)$. We therefore obtain
\begin{align}
C_{D}^{\max}(\mathsf{A})&\geq
\min_{\sum_{x}q_{x}=1}D\!\left(V^{\dagger}\mathsf{A}\,V,\{q_{x}I\}\right) \\
&=\min_{\sum_{x}q_{x}=1}D\!\left(\mathsf{A},\{q_{x}I\}\right)
=S_{D}(\mathsf{A}),
\end{align}
which together with the upper bound proves Eq.~\eqref{eq:MC_simplified}.
\end{proof}

The conditions that all $A_x$ share a common eigenbasis covers in particular all rank-one PVMs, and any POVM obtained by classical post-processing of a fixed PVM. The result extends, for distances that are additive in the outcomes, to a strictly larger class of POVMs whose elements need not share an eigenbasis but admit a common MUB structure:

\begin{corollary}\label{cor:MCcorol}
Let $\mathsf{A}=\{A_x\}_{x=1}^n$ be a POVM whose elements admit eigenbases $\{|\lambda_{i|x}\rangle\}_{i}$ that share a common MUB $\{|\lambda_{j}^{+}\rangle\}$ relative to the incoherent basis, i.e., $|\langle\lambda_{i|x}|\lambda_{j}^{+}\rangle|=1/\sqrt{d}$ for all $x,i,j$. If furthermore the distance is element-additive, $D(\mathsf{A},\mathsf{B})=\sum_{x}D(A_x,B_x)$, then
\begin{equation}\label{eq:MC_simplified_cor}
C_{D}^{\max}(\mathsf{A})\;=\;S_D(\mathsf{A}).
\end{equation}
\end{corollary}

\begin{proof}
The upper bound is established as in Theorem~\ref{thm:MCsharp}.

For the lower bound, choose the unitary $V$ mapping the incoherent basis to
the common MUB $\{|\lambda_{j}^{+}\rangle\}$. By unitary invariance of $D$,
\begin{align}
C_{D}^{\max}(\mathsf{A})
&\geq\min_{\mathsf{B}\in\mathcal{I}(d,n)}
D(V^{\dagger}\mathsf{A}\,V,\mathsf{B})\\
&=\min_{\mathsf{B}\in\mathcal{I}(d,n)}
D(\mathsf{A},V\mathsf{B}V^{\dagger})\\
&=\min_{\mathsf{B}\in\mathcal{I}(d,n)}
\sum_{x}D(A_x,VB_xV^{\dagger}).
\end{align}
For each $x$ apply the dephasing CP unital map
$\mathcal{E}_x^{\dagger}(X)=\sum_i\langle\lambda_{i|x}|X|\lambda_{i|x}\rangle\,
|\lambda_{i|x}\rangle\langle\lambda_{i|x}|$ corresponding to the
eigenbasis of $A_x$. By contractivity of $D$,
\begin{align}
D(A_x,VB_xV^{\dagger})\geq
D\!\left(\mathcal{E}_x^{\dagger}(A_x),\mathcal{E}_x^{\dagger}(VB_xV^{\dagger})\right).
\end{align}
Now $\mathcal{E}_x^{\dagger}(A_x)=A_x$, while $VB_xV^{\dagger}$ is a diagonal
operator in the MUB $\{|\lambda_{j}^{+}\rangle\}$ which (by the MUB
property) maps under $\mathcal{E}_x^{\dagger}$ to a multiple of the identity,
$\mathcal{E}_x^{\dagger}(VB_xV^{\dagger})=q_x I$, with
$q_x=\Tr[VB_xV^{\dagger}]/d=\Tr[B_x]/d$. Summing over $x$ then yields
\begin{align}
C_{D}^{\max}(\mathsf{A})\geq\min_{\sum_x q_x=1}\sum_x D(A_x,q_x I)
=S_D(\mathsf{A}),
\end{align}
which combined with the upper bound proves the corollary.
\end{proof}



A concrete and operationally meaningful family of distances satisfying the element-additivity condition in Corollary~\ref{cor:MCcorol} is the Schatten-$p$ norm distance introduced by Tendick \emph{et al.} in their distance-based quantification for measurement resources~\cite{Tendick2023}.  This distance is defined, for $p\in[1,\infty]$, by
\begin{equation}\label{eq:DpSinglePOVM}
D_{p}(\mathsf{A},\mathsf{B})
=\frac{1}{2}\sum_{x=1}^{n}\big\|A_{x}-B_{x}\big\|_{p},
\end{equation}
where $\|X\|_p=(\Tr|X|^p)^{1/p}$ denotes the Schatten $p$-norm. As shown in Theorem~5 of~\cite{Tendick2023}, the case $p=\infty$, corresponding to the spectral norm, defines a faithful and jointly convex distance that is contractive under CP unital maps, and hence satisfies all the requirements for a distance between measurements in the sense of \cite[Definition~1]{Tendick2023}. Moreover, the distance in \eqref{eq:DpSinglePOVM} is manifestly element-additive required in Corollary~\ref{cor:MCcorol}, since it can be written as
$D_{p}(\mathsf{A},\mathsf{B})=\sum_{x}D_{p}(A_{x},B_{x})$ with $D_{p}(A_{x},B_{x}):=\tfrac{1}{2}\|A_{x}-B_{x}\|_{p}$. It is also unitarily invariant, because $\|U(A_{x}-B_{x})U^{\dagger}\|_{p}=\|A_{x}-B_{x}\|_{p}$ for every unitary $U$ and every $p\in[1,\infty]$.
 
For $p=\infty$, all conditions of Corollary~\ref{cor:MCcorol} are therefore met, and the corollary implies the equivalence
$
C_{D_{\infty}}^{\max}(\mathsf{A})\;=\;S_{D_{\infty}}(\mathsf{A})$
whenever the elements of $\mathsf{A}$ admit a common MUB structure with respect to the incoherent basis. We note that Tendick \emph{et al.}~\cite{Tendick2023} explicitly observed that rank-$1$ projective measurements mutually unbiased to the incoherent basis maximize the
coherence among all rank-$1$ projective measurements, which is precisely the canonical instance of the common-MUB condition of Corollary~\ref{cor:MCcorol}. For such POVMs, the present equivalence identifies the maximal Schatten $\infty$-norm coherence of measurement with the Schatten $\infty$-norm sharpness, providing an explicit bridge between our results and the geometric quantifiers introduced in~\cite{Tendick2023}.

It is important to note that the equivalence $C_D^{\max}(\mathsf{A})=S_D(\mathsf{A})$ does not hold for arbitrary POVMs. When the elements of $\mathsf{A}$ have neither a common eigenbasis nor a common MUB structure (e.g., a SIC POVM in $d=2$, see Sec.~\ref{subsec:qubit}), the maximal coherence is in general smaller than the sharpness, i.e.,  $C_D^{\max}(\mathsf{A})\leq S_D(\mathsf{A})$ due to the inclusion $\mathcal{T}(d,n)\subset\mathcal{I}(d,n)$.

\subsection{Monotonicity under fuzzifying operations}

Even though the maximal coherence of measurements is not generally equivalent to the sharpness, it remains a natural question whether $C_{D}^{\max}$ can be considered as a legitimate monotone, i.e., a quantity that does not increase under any fuzzifying operation~\eqref{FO}, defined as the action of a CP unital map followed by convex mixing with a trivial measurement. Since $C_{D}^{\max}$ vanishes on trivial measurements (Theorem~\ref{thm:MCfaithful}) and is non-increasing under convex mixing by joint convexity, the question reduces to whether $C_{D}^{\max}$ is contractive under arbitrary CP unital maps. We address this question in two complementary cases: we first establish monotonicity for dichotomic quantum measurements in arbitrary dimension, and then for general cases under restricted but operationally relevant classes of channels, such as mixed-unitary channels and unitarily covariant channels.

\begin{corollary}[Monotonicity for dichotomic measurements]\label{cor:dichotomic_monotone}
Let $\mathsf{A}=\{A_{1},A_{2}\}$ be a dichotomic POVM on a $d$-dimensional Hilbert space ($d\geq 2$), and let $C_D^{\max}$ be the maximal coherence of measurement induced by a contractive, jointly convex, unitarily invariant distance $D$. Then $C_D^{\max}$ is monotonic under every fuzzifying operation~\eqref{FO}: for any CPTP channel $\mathcal{E}$ and any $\alpha\in[0,1]$, $\{q_{x}\}$ with $\sum_x q_x=1$,
\begin{equation}\label{eq:dichotomic_DPI}
C_{D}^{\max}\bigl(\mathcal{F}(\mathsf{A})\bigr)\;\leq\;C_{D}^{\max}(\mathsf{A}).
\end{equation}
\end{corollary}

\begin{proof}
The two elements of a dichotomic POVM satisfy $A_{2}=I-A_{1}$ and hence commute, $[A_{1},A_{2}]=0$. Consequently $A_{1}$ and $A_{2}$ share a common eigenbasis, and thus, $C_{D}^{\max}(\mathsf{A})=S_{D}(\mathsf{A})$ according to Theorem~\ref{thm:MCsharp}. The same argument applies to $\mathcal{F}(\mathsf{A})=\{\alpha\mathcal{E}^{\dagger}(A_{1})+(1-\alpha)q_{1}I_{B},\ \alpha\mathcal{E}^{\dagger}(A_{2})+(1-\alpha)q_{2}I_{B}\}$: completeness $\sum_{x}\mathcal{F}(A_{x})=I_{B}$ together with $\mathcal{F}$ having only two outcomes implies that the two elements of $\mathcal{F}(\mathsf{A})$ also commute and share a common eigenbasis, so $C_{D}^{\max}(\mathcal{F}(\mathsf{A}))=S_{D}(\mathcal{F}(\mathsf{A}))$. Combining these two identities with the monotonicity of $S_{D}$ under fuzzifying operations,
\[
C_{D}^{\max}\bigl(\mathcal{F}(\mathsf{A})\bigr)
\;=\;S_{D}\bigl(\mathcal{F}(\mathsf{A})\bigr)
\;\leq\;S_{D}(\mathsf{A})
\;=\;C_{D}^{\max}(\mathsf{A}),
\]
which proves Eq.~\eqref{eq:dichotomic_DPI}.
\end{proof}

Corollary~\ref{cor:dichotomic_monotone} shows that the equivalence $C_{D}^{\max}=S_{D}$ of Theorem~\ref{thm:MCsharp} implies monotonicity for dichotomic POVMs, since fuzzifying operations preserve the number of outcomes. For non-dichotomic POVMs, however, this equivalence is generally unavailable, and monotonicity must be addressed at the channel level, as we show below for the class of mixed-unitary channels.

\begin{theorem}\label{thm:MCmonotone}
Let $C_D$ be a distance-based coherence monotone induced by a jointly convex distance $D$. Then for any quantum channel $\mathcal{E}$ whose dual is a \emph{mixed-unitary} map, $\mathcal{E}^{\dagger}(\cdot)=\sum_{k}p_{k}W_{k}^{\dagger}(\cdot)W_{k}$ with $\{p_k\}$ a probability distribution and $\{W_k\}$ unitary, the maximal coherence of measurement satisfies the monotonicity
\begin{equation}
C_{D}^{\max}\big(\mathcal{E}^{\dagger}(\mathsf{A})\big)\;\leq\;C_{D}^{\max}(\mathsf{A}).
\end{equation}
\end{theorem}

\begin{proof}
By definition,
\begin{align}
C_{D}^{\max}\big(\mathcal{E}^{\dagger}(\mathsf{A})\big)
&=\max_{U}C_{D}\!\left(U^{\dagger}\sum_{k}p_{k}W_{k}^{\dagger}\mathsf{A}W_{k}\,U\right)
\\&=\max_{U}C_{D}\!\left(\sum_{k}p_{k}V_{k}^{\dagger}\mathsf{A}\,V_{k}\right),
\end{align}
where $V_{k}\equiv W_{k}U$. Since $C_D$ is induced by a jointly convex distance, it is jointly convex (see Appendix~\ref{App:1}, Lemma~\ref{join_convex}) and hence convex, so
\begin{align}
C_{D}\!\left(\sum_{k}p_{k}V_{k}^{\dagger}\mathsf{A}\,V_{k}\right)
&\leq\sum_{k}p_{k}\,C_{D}\!\left(V_{k}^{\dagger}\mathsf{A}\,V_{k}\right)\\
&\leq\sum_{k}p_{k}\,\max_{V}C_{D}\!\left(V^{\dagger}\mathsf{A}V\right)\\
&=C_{D}^{\max}(\mathsf{A}),
\end{align}
which holds for every $U$ and thus for the maximum.
\end{proof}

The same convexity argument extends, with only minor modifications, to unitarily covariant channels. If $\mathcal{E}^{\dagger}$ commutes with a group action under which $D$ is invariant, the contractivity of $D$ and joint convexity of $C_D$ yield $C_{D}^{\max}(\mathcal{E}^{\dagger}(\mathsf{A}))\leq C_{D}^{\max}(\mathsf{A})$ by the proof of Theorem~\ref{thm:MCmonotone}. The Werner--Holevo channel~\cite{Werner2002} can be included in this class.

\begin{corollary}\label{thm:covariant_monotone}
Let $C_D$ be a distance-based coherence monotone induced by a unitarily invariant distance $D$, and let $\mathcal{E}$ be a CPTP channel that is covariant under a unitary group $G$, i.e., $\mathcal{E}(U_g \cdot U_g^\dagger) = U_g \mathcal{E}(\cdot) U_g^\dagger$ for all $g \in G$. If the unitary $U^*$ maximizing $C_D(U^\dagger \mathcal{E}^\dagger(\mathsf{A}) U)$ can be chosen within $G$, then $C_D^{\max}(\mathcal{E}^\dagger(\mathsf{A})) \leq C_D^{\max}(\mathsf{A}).$
\end{corollary}

However, monotonicity under general fuzzifyng operations is structurally difficult due to the nested optimization in the definition of $C_{D}^{\max}$. The maximization over unitary conjugations and the minimization over incoherent measurements cannot, in general, be interchanged,
\[
\max_{U}\min_{\mathsf{B}\in\mathcal{I}(d,n)}D(U^{\dagger}\mathsf{A}\,U,\mathsf{B})\;\neq\;\min_{\mathsf{B}\in\mathcal{I}(d,n)}\max_{U}D(U^{\dagger}\mathsf{A}\,U,\mathsf{B}).
\]
A minimax interchange of this form would require, for instance, the convexity--concavity conditions of Sion's minimax theorem, which are not satisfied here. In this case, the set of unitaries is not convex, and the objective function is not generally concave in the unitary variable. Hence, for general fuzzifyng operations, it remains open whether $C_{D}^{\max}$ can be considered as a monotone in the resource theory of sharpness.

\begin{table}[]
\begin{tabular}{c|c|c}
\hline\hline
\multicolumn{2}{c|}{POVM and channel class} & Monotonicity \\
\hline
Dichotomic & any fuzzifyng operations & \checkmark~(Cor.~\ref{cor:dichotomic_monotone}) \\
\hline
\multirow{3}{*}{Non-dichotomic}
 & mixed-unitary channel        & \checkmark~(Thm.~\ref{thm:MCmonotone}) \\
 & unitarily covariant channel  & \checkmark~(Cor.~\ref{thm:covariant_monotone}) \\
 & any fuzzifyng operations    & Open \\
\hline\hline
\end{tabular}
\caption{Monotonicity of $C_{D}^{\max}$ under fuzzifying operations, classified by the structure of the input POVM and the underlying CP unital map. For dichotomic POVMs, monotonicity holds under any fuzzifying operation. For non-dichotomic POVMs, monotonicity is established under mixed-unitary channels (Theorem~\ref{thm:MCmonotone}) and unitarily covariant channels (Corollary~\ref{thm:covariant_monotone}), while the case of a general fuzzifyng operations remains open.}
\label{tab:monotonicity}
\end{table}

Table~\ref{tab:monotonicity} summarizes the monotonicity results according to the structure of the input POVM and the underlying CP unital map. For dichotomic POVMs ($n=2$), the two elements $\{A_{1},A_{2}=I-A_{1}\}$ commute and share a common eigenbasis in any dimension, so Theorem~\ref{thm:MCsharp} yields $C_{D}^{\max}=S_{D}$ implying monotonicity under any fuzzifying operations, as shown in Corollary~\ref{cor:dichotomic_monotone}. For non-dichotomic POVMs, the equivalence with sharpness is generally unavailable, and monotonicity must be considered at the channel level: it holds under mixed-unitary channels (Theorem~\ref{thm:MCmonotone}) and under unitarily covariant channels (Corollary~\ref{thm:covariant_monotone}), but whether $C_{D}^{\max}$ remains monotonic under general CP unital maps in this regime is left as an open problem.

\section{Examples}\label{sec:examples}

\subsection{Qubit POVMs: equivalence and its breakdown}\label{subsec:qubit}

\begin{figure*}[t]
    \centering
    \includegraphics[width=0.9\textwidth]{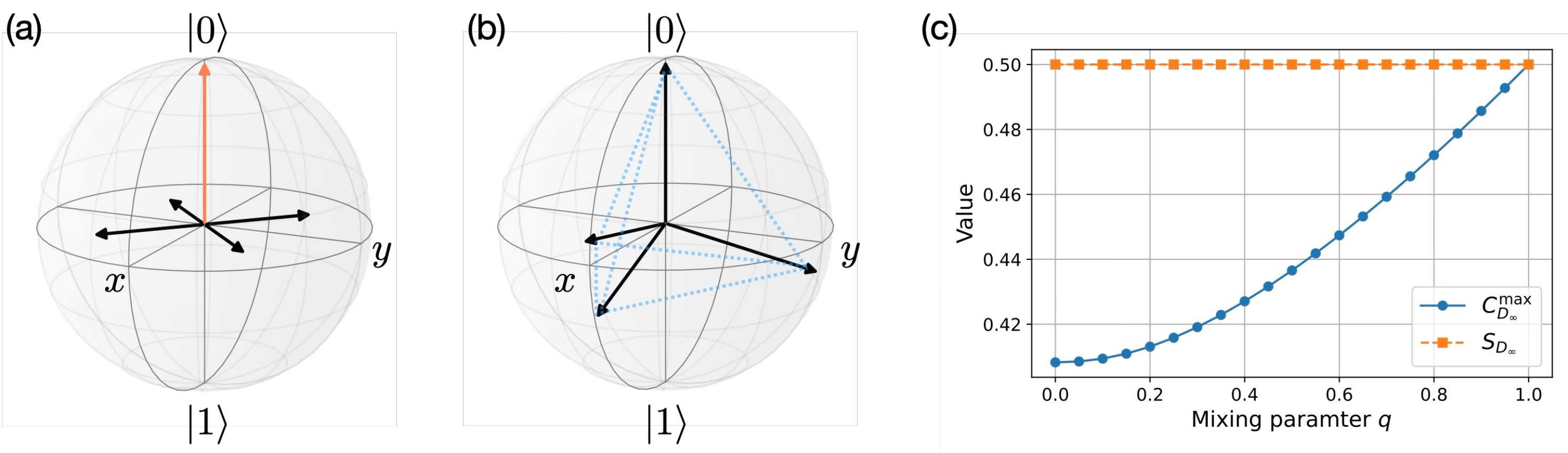}
    \caption{Comparison of the sharpness $S_D$ and the maximal coherence of measurement $C_D^{\max}$. 
    (a) Bloch-sphere representation of the joint $X$--$Y$ POVM $\mathsf{M}_{XY}$;  all four rank-one elements have eigenbases lying in the $xy$-plane of the Bloch sphere; the incoherent basis $\{|0\rangle,|1\rangle\}$ along $z$-axis (orange arrows) is mutually unbiased to every such eigenbasis, so the common-MUB condition of Corollary~\ref{cor:MCcorol} is satisfied.    
    (b) Bloch-sphere representation of a tetrahedral SIC qubit POVM $\mathsf{M}_{\rm SIC}$; its four rank-one elements have neither a common eigenbasis nor a common MUB. 
    (c) $S_D$ and $C_D^{\max}$ as a function of the mixing parameter $p$ for the convex combination $q\,\mathsf{M}_{XY}+(1-q)\,\mathsf{M}_{\rm SIC}$. The two quantities coincide at $q=1$, in agreement with Theorem~\ref{thm:MCsharp}, and differ for $q\neq 1$, illustrating the breakdown of the equivalence when the common-MUB condition fails.}
    \label{FigQubit}
\end{figure*}

To illustrate the role of common-MUB conditions in the relationship between maximal coherence of measurements and sharpness, we compare two qubit POVMs: the joint $X$--$Y$ POVM
\[
\mathsf{M}_{XY}=\big\{\tfrac{1}{4}\big[I+\tfrac{1}{\sqrt{2}}((-1)^i\sigma_x
+(-1)^j\sigma_y)\big]\big\}_{i,j=0,1},
\]
and the symmetric, informationally complete (SIC) tetrahedral qubit POVM
\[
\mathsf{M}_{\rm SIC}=\big\{\tfrac{1}{4}(I+\vec{n}_i\cdot\vec{\sigma})\big\}_{i=1}^{4},
\]
where $\{\vec{n}_i\}$ point to the vertices of a tetrahedron in the Bloch sphere. Both are illustrated in Fig.~\ref{FigQubit}(a) and (b), respectively.

Since we are focusing on the common-MUB condition of Corollary~\ref{cor:MCcorol}, the natural choice of distance is one that is \emph{element-additive}. We therefore adopt, throughout this example, the Schatten $\infty$-norm (spectral norm), which is the $p=\infty$ instance of Eq.~\eqref{eq:DpSinglePOVM}. The induced maximal coherence and sharpness monotones, $C_{D_{\infty}}^{\max}$ and $S_{D_{\infty}}$, are the quantities plotted in Fig.~\ref{FigQubit}(c).

The four rank-one elements of $\mathsf{M}_{XY}$ have eigenbases lying in the $xy$-plane of the Bloch sphere, and the incoherent basis $\{|0\rangle,|1\rangle\}$ along $\hat z$ (the orange arrows in Fig.~\ref{FigQubit}(a)) is mutually unbiased to every such basis. The common-MUB condition of Corollary~\ref{cor:MCcorol} is therefore
satisfied, and we expect $C_{D_{\infty}}^{\max}(\mathsf{M}_{XY})=S_{D_{\infty}}(\mathsf{M}_{XY})$. By contrast, the eigenvectors of the four rank-one elements of $\mathsf{M}_{\rm SIC}$ point to the vertices of a regular tetrahedron and admit no common MUB relative to the incoherent basis, so neither Theorem~\ref{thm:MCsharp} nor Corollary~\ref{cor:MCcorol} applies, and the equivalence is expected to fail.

We verify this prediction by evaluating $C_{D_{\infty}}^{\max}$ and $S_{D_{\infty}}$ along the convex path $\mathsf{M}_{q}=q\,\mathsf{M}_{XY}+(1-q)\,\mathsf{M}_{\rm SIC}$, $q\in[0,1]$. As shown in Fig.~\ref{FigQubit}(c), the two quantities agree exactly at $q=1$, in agreement with Corollary~\ref{cor:MCcorol}, and exhibit a strictly positive gap for $q<1$, illustrating that the common-MUB structure of the POVM elements is essential for the
equivalence to hold.

\subsection{Operational interpretation: connection to Fisher information}

\begin{figure}[ht!]
    \centering
    \includegraphics[width=0.5\textwidth]{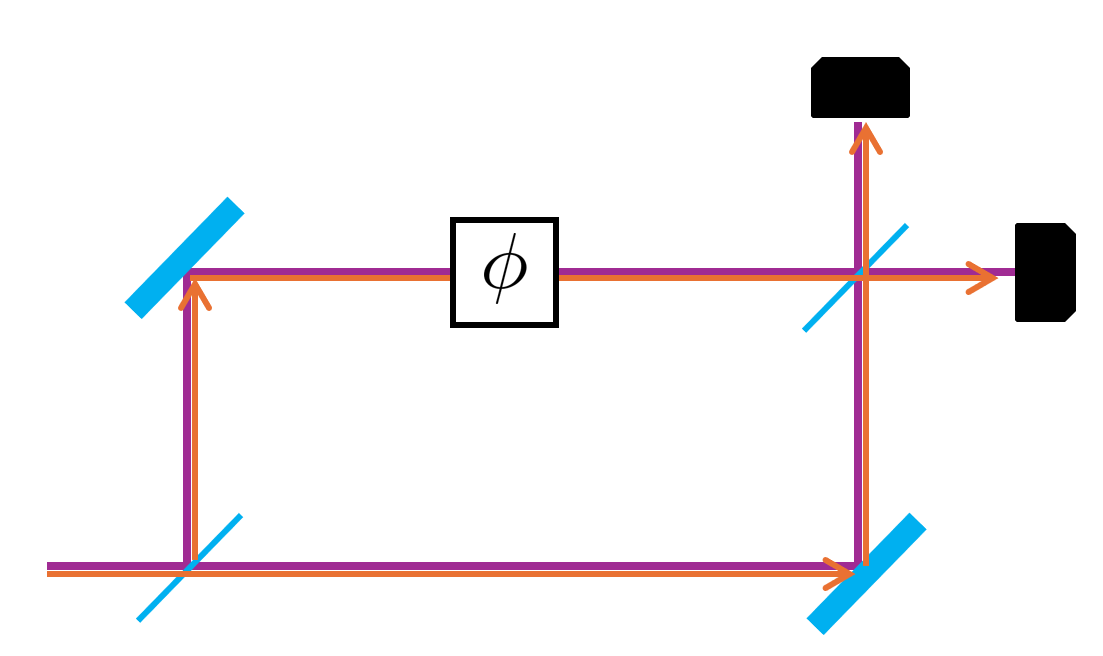}
    \caption{Schematic of a single-photon Mach--Zehnder phase-sensing
    experiment. The unknown phase $\phi$ is imprinted on one arm; an
    imperfect interferometer (visibility $\nu$) realizes the
    binary POVM~\eqref{eq:PEPOVM}.}\label{FigPSEExample}
\end{figure}

Several works have related quantum-information-theoretic resources to operationally meaningful figures of merit, such as Fisher information~\cite{Yuan2015,Ma2019}. We now provide a simple, fully analytical example connecting our maximal coherence of measurement to the Fisher information in single-photon phase sensing.

The setup is illustrated in Fig.~\ref{FigPSEExample}. A single photon is injected into a Mach--Zehnder interferometer, one arm of which contains a controllable phase shifter $\phi$. After recombination at the second beam splitter, the photon exits through one of the two output ports and is detected. In the presence of imperfections such as mode mismatch or detector inefficiencies, the binary measurement realized at the output is described by the POVM $\mathsf{E}(\phi) = \{E_+(\phi), E_-(\phi)\}$ with
\begin{eqnarray}\label{eq:PEPOVM}
E_{+}(\phi)&=&\frac{1}{2}\big[I+\nu(\cos\phi\,\sigma_{x}
+\sin\phi\,\sigma_{y})\big],\\
E_{-}(\phi)&=&I-E_{+}(\phi),
\end{eqnarray}
where $\nu\in[0,1]$ is the fringe visibility, summarizing all imperfections that reduce interference contrast (timing or polarization mismatch, dark counts, finite quantum efficiency, etc.). On the input state $|+\rangle\langle+|$, the detection statistics are
\begin{align}
p_{+}(\phi)=\tfrac{1}{2}(1+\nu\cos\phi),\qquad
p_{-}(\phi)=1-p_{+}(\phi).
\end{align}

The classical Fisher information of this binary distribution is
\begin{align}
\mathcal{I}(\phi)
=\frac{[\partial_{\phi}p_{+}(\phi)]^{2}}{p_{+}(\phi)\,p_{-}(\phi)}
=\frac{\nu^{2}\sin^{2}\phi}{1-\nu^{2}\cos^{2}\phi},
\end{align}
which attains its maximum at $\phi=\pi/2$, yielding
\begin{align}
\mathcal{I}_{\max}=\nu^{2}.
\end{align}

Now consider the same setup from the resource-theoretic perspective. Both elements of \eqref{eq:PEPOVM} share the common eigenbasis along the direction $\mathbf{m}(\phi)=(\cos\phi,\sin\phi,0)$, so we have $C_{D_1}^{\max}=S_{D_1}$ by Theorem~\ref{thm:MCsharp} , where $D_1$ is the trace-distance based maximal statistical distance. A direct computation gives
\begin{align}
C_{D_1}^{\max}(\mathsf{E}(\phi))
=\min_{q\in[0,1]}\sup_{|\mathbf{s}|\leq 1}
\frac{1}{2}\big|1+\nu\,\mathbf{m}\cdot\mathbf{s}-2q\big|
=\frac{\nu}{2},
\end{align}
attained at $q=1/2$ and $\mathbf{s}$ aligned with $\pm\mathbf{m}$. The result is independent of $\phi$, as expected from the unitary invariance of the maximal coherence.

Although $C_{D_1}^{\max}$ above is defined via the trace-distance based maximal statistical distance $D_1$, the POVM $\{E_{+}(\phi),E_{-}(\phi)\}$ is dichotomic, and Appendix~I of~\cite{Tendick2023} guarantees that for dichotomic POVMs the trace-distance based resource quantification coincides with the diamond-norm based one, together with the Schatten $\infty$-norm based distance $D_{\infty}$. The induced maximal coherence monotones therefore satisfy $C_{D_{1}}^{\max}=C_{D_{\diamond}}^{\max}=C_{D_{\infty}}^{\max}=\nu/2$, and the identity above can equivalently be written as
\begin{align}
\mathcal{I}_{\max}\;=\;\big(2\,C_{D_{1}}^{\max}\big)^{2}.
\end{align}
The maximal Fisher information therefore coincides with the squared maximal coherence of measurement, equivalently with the sharpness in any of these mutually equivalent formulations, providing a direct operational interpretation of $C_{D}^{\max}$ as quantifying the achievable phase sensitivity of the device. In the language of resource theories: a sharper measurement enables better metrological performance, and the relation is quantitatively saturated in this paradigmatic single-photon setup.

\subsection{Comparison of sharpness measures: imperfect PNRD}

We finally compare our maximal coherence of measurement with several previously proposed unsharpness measures in a physically motivated example: the photon-number resolving detector (PNRD) with quantum efficiency $\eta\in[0,1]$, whose POVM elements are~\cite{Achilles2003}
\begin{align}\label{eq:PNRDPOVM}
E_{k}^{(\eta)}=\sum_{n=k}^{\infty}\binom{n}{k}\eta^{k}(1-\eta)^{n-k}
|n\rangle\langle n|.
\end{align}

We evaluate six sharpness measures on this family, grouped according to the underlying notion of sharpness they implement. The first group comprises measures of PVM-unsharpness: the Hilbert--Schmidt (HS) norm $\|\mathsf{A}\|_{\rm HS}^{2}=\sum_{x}\Tr(A_x^{2})$, motivated by analogy with state purity; the entropic measure $U_{\rm ent}$~\cite{Baek2016}; the variance-based measure $U_{\rm var}$~\cite{Liu2021}; and the operator-norm based measure $U_{\rm op}$~\cite{Mitra2022}. The second group comprises monotones for the generalized sharpness of Definition~\ref{def:gensharp}: our maximal coherence of measurement $C_{D_{\diamond}}^{\max}$ and the tunability-based measure of Buscemi \emph{et al.}~\cite{Buscemi2024}, $\kappa_{u}^{\ast}(\mathsf{A}\|Z)=\max_{\mathcal{F}}\kappa_{u}(\mathcal{F}(\mathsf{A}):Z)$, where $Z$ is a reference POVM and the optimization is over fuzzifying operations $\mathcal{F}$.

\begin{figure}[ht!]
    \centering
    \includegraphics[width=0.45\textwidth]{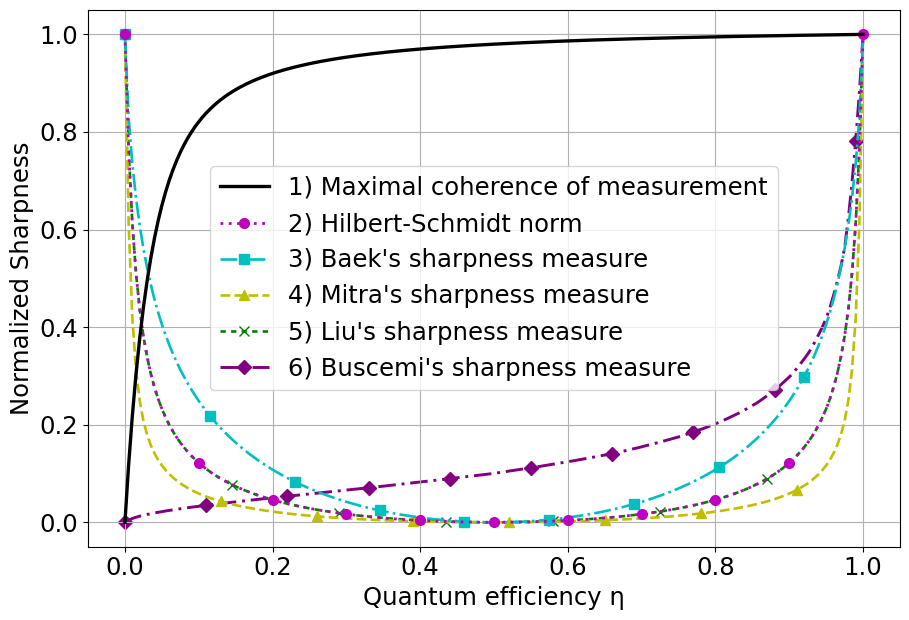}
    \caption{Comparison of sharpness measures for the imperfect PNRD as a   function of quantum efficiency $\eta$. All measures are normalized to
    the unit interval, and unsharpness measures are converted to sharpness
    via $1-\mathcal{US}$. The HS norm and the measures of Baek, Mitra, and
    Liu--Luo are non-monotonic in $\eta$, reaching a maximum at both
    $\eta=0$ (vacuum-projecting trivial measurement) and $\eta=1$ (ideal
    PNRD). Only Buscemi \emph{et al.}'s measure and the maximal coherence
    of measurement are monotonic, vanishing on the trivial measurement and
    increasing to the maximum on the ideal PNRD, hence both qualify as
    valid sharpness monotones in the resource theory of
    sharpness.}\label{FigPNRD}
\end{figure}

At $\eta=1$ the PNRD reduces to an ideal photon-number measurement, which is a PVM on the Fock basis. At the opposite limit $\eta=0$, the POVM elements collapse to $E_{0}^{(0)}=\mathbb{1}$ and $E_{k}^{(0)}=0$ for $k\geq 1$, i.e.\ the trivial measurement that always reports the outcome ``zero photons'' regardless of the input state; this trivial POVM is itself formally a PVM, with rank-zero elements for $k\geq 1$. The HS norm and the unsharpness measures $U_{\rm ent}$~\cite{Baek2016}, $U_{\rm var}$~\cite{Liu2021}, and $U_{\rm op}$~\cite{Mitra2022} all detect deviation from a PVM, so they are maximized at both $\eta=0$ and $\eta=1$ and hence non-monotonic in $\eta$.

This non-monotonicity exposes an operationally important difference between the PVM-based and the generalized notions of sharpness: in any resource theory built on the latter, the trivial measurement (here at $\eta=0$) must be a free object, while under the PVM-based notion it is indistinguishable from the ideal PNRD. By contrast, our maximal coherence of measurement $C_{D_{\diamond}}^{\max}$ and the tunability-based measure of Buscemi \emph{et al.}~\cite{Buscemi2024} both vanish on the trivial POVM at $\eta=0$ and increase monotonically with $\eta$ to attain their maximum on the ideal PNRD. Since all PNRD elements~\eqref{eq:PNRDPOVM} are diagonal in the Fock basis, we have $C_{D_{\diamond}}^{\max}(\mathsf{E}^{(\eta)}) =S_{D_{\diamond}}(\mathsf{E}^{(\eta)})$ from Theorem~\ref{thm:MCsharp}, so the monotonicity of $C_{D_{\diamond}}^{\max}$ follows directly from the monotonicity of sharpness. The maximal coherence of measurement is therefore operationally aligned with the generalized sharpness rather than with the PVM-based notion.

\section{Conclusion}\label{sec:conclusion}

We have investigated the relation between coherence and sharpness of quantum measurements from a resource-theoretic perspective. Measurement coherence is defined relative to a fixed reference basis, whereas measurement sharpness is basis-independent and quantifies the deviation of a POVM from trivial, state-independent measurements. Our main result shows that these two notions are connected by maximizing the coherence of a measurement over all unitary choices of basis.

For POVMs whose elements share a common eigenbasis, we proved that the maximal distance-based coherence of measurement coincides exactly with the corresponding distance-based sharpness monotone. We further extended this equivalence, for element-additive distances, to POVMs whose elements admit a common mutually unbiased basis structure. These results provide a measurement-level analogue of the maximal-coherence/purity relation for quantum states, identifying maximal coherence as a basis-independent quantifier of measurement sharpness in these regimes.

We also studied the monotonicity of maximal coherence under fuzzifying operations. Faithfulness holds in full generality: the maximal coherence vanishes if and only if the measurement is trivial. Monotonicity is established for dichotomic POVMs, as well as for non-dichotomic POVMs under mixed-unitary preprocessing channels and unitarily covariant channels. For arbitrary CP unital preprocessing of general non-dichotomic POVMs, however, monotonicity remains open due to the nested optimization over unitary bases and incoherent measurements.

Finally, we illustrated the scope and limitations of the equivalence through several examples. The joint $X$--$Y$ qubit POVM satisfies the common-MUB condition and hence realizes the equivalence, while the tetrahedral SIC POVM shows that the equivalence can fail without such structure. In a Mach--Zehnder phase-sensing setup, the maximal Fisher information is determined by the squared maximal coherence of the measurement, giving an operational interpretation of the quantity. In a noisy photon-number resolving detector, maximal coherence behaves as a proper sharpness monotone, unlike conventional PVM-based unsharpness measures. 
Future work includes characterizing the gap between maximal coherence and sharpness for general POVMs and resolving monotonicity under arbitrary fuzzifying operations. Another natural direction is to extend the present framework toward a unified resource-theoretic framework of measurement resources, including coherence, sharpness, and incompatibility.

\begin{acknowledgments}
  This work was supported by the National Research Foundation of Korea(NRF) grant funded by the Korea government(MSIT) (No. RS-2025-18362970, RS-2026-25519864), the Korea Institute of Science and Technology Information (KISTI) (Grant No. P25027), the Korean ARPA-H Project through the Korea Health Industry Development Institute (KHIDI) funded by the Ministry of Health \& Welfare, Republic of Korea (RS-2025-25456722), and the Ministry of Trade, Industry, and Energy (MOTIE, Korea, under the project “Industrial Technology Infrastructure Program” (RS2024-00466693).
\end{acknowledgments}

\appendix

\section{Faithfulness and monotonicity of distance-based quantification}\label{App:1}

\begin{lemma}[Joint convexity of induced statistical distances]
\label{join_convex}
If a statistical distance $D$ is jointly convex in its arguments, then the
maximal statistical distance defined in Eq.~\eqref{eq:Dstat} is also
jointly convex on POVMs.
\end{lemma}

\begin{proof}
Joint convexity of $D$ on probability distributions states that, for any
probability distribution $\{\alpha_i\}$ and any state $\rho$,
\begin{equation}
D\!\left(\sum_i\alpha_i\mathbf{p}_{\rho,\mathsf{A}_i},
\sum_i\alpha_i\mathbf{p}_{\rho,\mathsf{B}_i}\right)
\leq\sum_i\alpha_i D(\mathbf{p}_{\rho,\mathsf{A}_i},
\mathbf{p}_{\rho,\mathsf{B}_i}).
\end{equation}
Then, denoting $\mathsf{A}=\sum_i\alpha_i\mathsf{A}_i$,
$\mathsf{B}=\sum_i\alpha_i\mathsf{B}_i$,
\begin{align}
D(\mathsf{A},\mathsf{B})
&=\max_{\rho}D\!\left(\sum_i\alpha_i\mathbf{p}_{\rho,\mathsf{A}_i},
\sum_i\alpha_i\mathbf{p}_{\rho,\mathsf{B}_i}\right)\nonumber\\
&\leq\max_{\rho}\sum_i\alpha_i D(\mathbf{p}_{\rho,\mathsf{A}_i},
\mathbf{p}_{\rho,\mathsf{B}_i})\nonumber\\
&\leq\sum_i\alpha_i\max_{\rho}
D(\mathbf{p}_{\rho,\mathsf{A}_i},\mathbf{p}_{\rho,\mathsf{B}_i})\\
&=\sum_i\alpha_i D(\mathsf{A}_i,\mathsf{B}_i).
\end{align}
\end{proof}

\begin{lemma}[Monotonicity of distance-based sharpness]
\label{lem:sharp_mono}
The distance-based sharpness defined in Eq.~\eqref{eq:D_sharp} is
monotonic under fuzzifying operations,
$S_D(\mathsf{A})\geq S_D(\mathcal{F}(\mathsf{A}))$, for every fuzzifying
operation $\mathcal{F}$, provided $D$ is contractive under CP unital maps
and jointly convex.
\end{lemma}

\begin{proof}
Write a fuzzifying operation as
$\mathcal{F}(\mathsf{A})=\alpha\,\mathcal{E}^{\dagger}(\mathsf{A})
+(1-\alpha)\{q_xI\}$ as in Eq.~\eqref{FO}. Then
\begin{align}
S_D(\mathsf{A})&=\min_{\sum_xp_x=1}D(\mathsf{A},\{p_xI\})\\
&\geq\min_{\sum_xp_x=1}D\!\left(\mathcal{E}^{\dagger}(\mathsf{A}),
\mathcal{E}^{\dagger}(\{p_xI\})\right)\\
&=\min_{\sum_xp_x=1}D\!\left(\mathcal{E}^{\dagger}(\mathsf{A}),\{p_xI\}\right)
=S_D(\mathcal{E}^{\dagger}(\mathsf{A})),
\end{align}
where the inequality is the contractivity of $D$ and the second equality
uses unitality of $\mathcal{E}^{\dagger}$. Joint convexity then gives
\begin{align}
&S_D(\mathcal{E}^{\dagger}(\mathsf{A}))\\
&=\min_{\sum_xp_x=1}D(\mathcal{E}^{\dagger}(\mathsf{A}),\{p_xI\})\\
&\geq\min_{\sum_xp_x=1}D\!\left(\alpha\mathcal{E}^{\dagger}(\mathsf{A})
+(1-\alpha)\{q_xI\},\,\alpha\{p_xI\}+(1-\alpha)\{q_xI\}\right)\\
&=\min_{\sum_xp'_x=1}D(\mathcal{F}(\mathsf{A}),\{p'_xI\})
=S_D(\mathcal{F}(\mathsf{A})),
\end{align}
since the convex combination of trivial measurements is again a trivial
measurement, and the minimum over $\{p'_x\}$ ranges over all probability
distributions as $\{p_x\}$ does. Combining the two chains of inequalities
gives the claim.
\end{proof}

\bibliography{mybibfile}

\end{document}